\documentclass[a4paper,12pt]{extarticle}

\usepackage{titlesec}

\titleformat{\paragraph}
{\normalfont\normalsize\bfseries}{\theparagraph}{1em}{}
\titlespacing*{\paragraph}
{0pt}{3.25ex plus 1ex minus .2ex}{1.5ex plus .2ex}

\usepackage{hyperref}
\usepackage{amsmath}

\usepackage{enumerate}

\usepackage{amsthm}
\usepackage{amsfonts}
\usepackage{subeqnarray}
\usepackage{mathrsfs}
\usepackage{latexsym}
\usepackage{dsfont}
\usepackage{mathtools} 
\usepackage{upgreek}
\usepackage{stmaryrd}

\usepackage[margin=2cm]{geometry}
\allowdisplaybreaks[4]

\newcommand{\ud}{\mathrm{d}}
\newcommand{\beq}{\begin{equation}}
\newcommand{\eeq}{\end{equation}}
\newcommand{\ba}{\begin{eqnarray}}
\newcommand{\ea}{\end{eqnarray}}

\newcommand{\f}{\varphi}
\newcommand{\ms}{\mathscr}
\newcommand{\mb}{\mathbb}

\newcommand{\e}{\epsilon}
\newcommand{\g}{\kappa}

\newcommand\lb{\llbracket}
\newcommand\rb{\rrbracket}

\newtheorem{theorem}{Theorem}[section]
\newtheorem{remark}{Remark}[section]

\newtheorem{lemma}{Lemma}[section]

\title{New Temperature Dependent Configurational Probability Diffusion Equation For Diluted FENE Polymer Fluids: Existence of Solution Results}

\author{
Ionel Sorin Ciuperca $^1$ and Liviu Iulian Palade$^2$ 
\footnote{Corresponding author.  E-mail:  liviu-iulian.palade@insa-lyon.fr. }
\thanks{Dedicated to the memory of late Professor Genevi\`eve Raugel, Universit\'e Paris-Sud, in fond remembrance.
}   
}

\numberwithin{equation}{section}

\begin{document}

\maketitle

\begin{flushleft}

Universit\'e de Lyon, CNRS, Institut Camille Jordan UMR 5208\\

$^1$ Universit\'e Lyon 1,  B\^at Braconnier, 43 Boulevard du 11 Novembre 1918, F-69622, Villeurbanne, France. 

$^2$ INSA-Lyon,  P\^ole de Math\'ematiques, B\^at. Leonard de Vinci No. 401, 21 Avenue Jean Capelle, F-69621, Villeurbanne, France.

\end{flushleft}

\begin{abstract}
The theory for the non-isothermal rheology of polymer fluids proposed in ~\cite{bird3} used  several approximations including the so-called linear gradient approximations for the temperature field and Brownian forces.  While it had the significant advantage of dealing with linear equations, the approximations involved may have led to several non-physical predictions.  This work is a continuation of ~\cite{bird3} in that it obtains the corresponding non-linear configurational probability density  equation in dimensionless form without the linear gradient approximations for the temperature field and Brownian forces.  It does so for incompressible  diluted polymer solutions with polymer molecules being modeled as FENE (\textit{F}initely \textit{E}xtensible \textit{N}onlinear \textit{E}lastic) chains.  Next we prove the existence of temperature dependent, positive variational solutions for the probability density equation of the FENE model.


\end{abstract}

\begin{flushleft}

 


\end{flushleft}


\section{Introduction}\label{intro}

All polymeric liquid flows of practical importance and in everyday life - e.g. injection molding, hydrodynamics of biological fluids, thin film flows, lubrication (to name only a very few) - are subjected to significant heat transfer  processes.  Therefore all these flows are non-isothermal and require to be studied as such (see e.g. the recent work ~\cite{fhy1}, ~\cite{mp1} and references cited therein).  However, such an undertaking is rendered more complicated by the necessity to focus, in addition to the momentum balance and fluid constitutive equations, on the temperature equation as well.

In the kinetic theory of polymer dynamics one is interested in producing constitutive equations by taking into account the dominant interactions between fluid constituents (polymer-polymer, polymer-solvent molecules) which ultimately govern the macroscopic physical properties.  In doing so, a configurational (conformation) probability  diffusion equation - hereafter referred to  as the CPD equation - is obtained; its solution, denoted $\Psi $, enters both the expression of the momentum balance (via the stress tensor expression) and the temperature balance equation. If the temperature is not held constant, all three aforementioned governing equations are interrelated, forming a system of equations. 

Curtiss and Bird undertook in ~\cite{bird3} to extend existing isothermal polymer kinetic theory results to non-isothermal flows.  In order to present simpler forms for the balance law equations, the velocity, temperature and concentration fields are approximated via first order truncated Taylor expansions (higher order derivatives being deemed as negligible), hence the so-called linear gradients approximations; this assumption is listed as the 5th simplifying hypothesis, see page 85 of Section 17 in ~\cite{bird3}.  Consequently, the temperature is expressed as a first order approximation in equation (12.3) on page 49 in ~\cite{bird3}, and its impact on the Brownian force appears in equation (12.16) on page 53 of the same.  The theory's appealing elegance and acclamation notwithstanding, some of the approximations involved (like the linear gradients approximation) may have led to some nonphysical predictions, such as infinite viscosity at a finite extensional rate or the absence of influence of temperature gradients  orthogonal to the direction of flow, as noted in ~\cite{bird3}.   As a matter of fact, the Authors themselves invite on page 85 of ~\cite{bird3} for further work as the ``...major assumptions can and should be challenged''. 

This paper is a first step and exploratory in nature continuation of the work in ~\cite{bird3} in that it focuses on the theory without using linear gradients approximations.  Specifically, it is devoted to: 

\begin{enumerate}[$\bullet$]
 \item obtaining the dimensionless form of the CPD equation without the linear gradient approximation for the temperature field and Brownian forces; for the later a classical, Boltzmann description, is assumed instead
 \item proving the existence of positive variational solutions to the \textit{F}initely \textit{E}xtensible \textit{N}onlinear \textit{E}lastic (FENE) model CPD equation using Schauder's fixed-point theory
\end{enumerate}

As to the second goal, the importance of studying the full system of equations consisting of momentum balance CPD (and of the closely related stress tensor) and the heat diffusion equations notwithstanding (see for instance  ~\cite{zz2}, ~\cite{jblo1}, ~\cite{bs2}, ~\cite{bs3}, ~\cite{bcip}), it is to be noted here that it is common matter in the rheology literature (as per ~\cite{bird2}, ~\cite{faith}, ~\cite{yhl}) to first focus on the CPD alone, see e.g. ~\cite{bcip}, ~\cite{ch1}, ~\cite{chp1}, ~\cite{chp2}, ~\cite{chpp1}, ~\cite{cp1} - ~\cite{cp4},  ~\cite{lip1}. 

The paper is organized as follows:

\begin{enumerate}[$\bullet$]
 \item Section \ref{p1} is devoted to obtaining the CPD equation in dimensionless form for the FENE  model  
\item Section \ref{p2} deals with proving the existence of variational solutions to the CPD equation of the FENE model
 \subitem Subsection \ref{e} introduces the mathematical problem under scrutiny
 \subitem Subsection \ref{exreg} gives the proof of the existence of positive solutions to the regularized problem 
 \subitem Subsection \ref{ee} gives estimates uniform in $\e$
 \subitem Subsection \ref{nge} summarizes the final results and gives the main existence result
\end{enumerate}

\section{Modeling The Non-Isothermal Dynamics: The CPD Equation For \textit{FENE} Polymer Chains }\label{p1}


The notations are akin to those of ~\cite{bird1} - ~\cite{bird2} (see also ~\cite{tig}).  The fluids are assumed incompressible. 

A polymer molecule in a diluted solution is here modeled as a FENE chain, that is as an elastic dumbbell with the polymer mass concentrated at its two extremities (see Chapter 13 in ~\cite{bird2} for a detailed description).  Let $x$ be the Eulerian (macroscopic) variable,  $Q$ the end-to-end microscopic vector, and $F^{(c)}$ the non-linear elastic recoil spring force. Denote by $\Psi=\Psi(x,Q,t)$ the configurational probability, where  $\nabla_y$  is the gradient in the $y$-direction.

Let $\g=\nabla_x v(x,t) $ be the macroscopic velocity gradient.  Moreover, $\g\cdot x= k_{ij}x_j $ by Einstein's convention for repeated indices. The solvent is here assumed to be a classical, with shear rate independent but temperature dependent $\eta_s$ viscosity Newtonian fluid.   

Let $\nu=1,2$ label each individual bead.  In the absence of external forces, the (vector) force balance equation for each bead reads: 

\beq\label{cd1}
F_{\nu}^{(\phi)}+F_{\nu}^{(h)}+F_{\nu}^{(b)}=0,~ \nu=1,2
\eeq

In the above,  $F_{\nu}^{(\phi)}$ is the intramolecular force that accounts for the polymer molecule entropic elasticity, here formally modeled by the \textit{FENE} spring and assumed to be a potential force (i.e. $\displaystyle  F_{\nu}^{(\phi)}=-\nabla_{r_\nu}\Phi $, where $\Phi $ is a given potential function and $r_\nu$ is the position vector of bead $\nu$).

Next, $F_{\nu}^{(h)} $ is the hydrodynamic drag force and $F_{\nu}^{(b)}$ is the Brownian force caused by thermal fluctuations (that pushes beads to jostle about randomly). Their expressions are:

\beq\label{cd2}
F_{\nu}^{(h)}=-\zeta \left[ \bigl\lb \stackrel{\cdot}{r_\nu}\bigr\rb-v_\nu \right],~ \nu=1,2 
\eeq

where $\zeta$ is the hydrodynamic drag coefficient, $\bigl\lb ~ \bigr\rb$ stands for a velocity-space average (as in equation 13.1-4 of ~\cite{bird2}), and

\beq\label{cd3}
F_{\nu}^{(b)}= -k_B  \nabla_{r_\nu}\left(T\left( r_\nu,t\right) \ln \Psi\left(x,Q,t \right)\right) ,~ \nu=1,2 
\eeq

with $k_B$ the Boltzmann's constant, and $T$ denoting the temperature.

Writing $v_\nu=v_0+\g\cdot  r_\nu  $ (basically a 1st order expansion of $v_\nu$ about $v_0$), with the help of \eqref{cd2}-\eqref{cd3} the force balance equation \eqref{cd1} reads:

\beq\label{cd4}
-\zeta \left[ \bigl\lb \stackrel{\cdot}{r_\nu}\bigr\rb -v_0-\g \cdot  r_\nu   \right] -k_B \nabla_{r_\nu}\left(T\left( r_\nu,t\right) \ln \Psi\left(x,Q,t \right)  \right) +F_{\nu}^{(\phi)}=0,~ \nu=1,2 
\eeq

Summing over $\nu$ in \eqref{cd4}, and because $F_{1}^{(\phi)}+F_{2}^{(\phi)}=0$, it leads to:

\begin{align}\label{cd5}
&   -   \zeta  \left[\bigl\lb \stackrel{\cdot}{r_1}+\stackrel{\cdot}{r_2} \bigr\rb -2 v_0 - \g \cdot\left( r_1+r_2\right)  \right] &  \nonumber\\ 
&  -  k_B  \left[    \nabla_{r_1}\left( T\left( r_1,t\right) \ln \Psi\left(x,Q,t \right) \right) + \nabla_{r_2}\left( T\left( r_2,t\right) \ln \Psi\left(x,Q,t \right) \right) \right]  =0  &                                 
\end{align}

We have that $\bigl\lb\stackrel{\cdot}{r_1}+\stackrel{\cdot}{r_2} \bigr\rb=2\bigl\lb\stackrel{\cdot}{r_c} \bigr\rb $, and $\g \cdot \left(r_1+r_2 \right)=2\g\cdot \left(  r_c \right) $.  In order to relate the microscopic molecular scale to the macroscopic fluid flow scale, we make the following homogenization assumption: that $r_c$ and the outer $x$ Eulerian variable are the same.  Now, since $r_1\left(r_c,Q \right)=r_c-\dfrac{1}{2}Q  $ and $r_2\left(r_c,Q \right)=r_c+\dfrac{1}{2}Q  $, then $\nabla_{r_1}= \dfrac{1}{2} \nabla_x-\nabla_Q$ and $\nabla_{r_2}= \dfrac{1}{2} \nabla_x+\nabla_Q$.  Moreover, we Taylor expand $T(r_{1,2},t) $ to get, respectively, $T(r_{1},t)=T\left( r_c-\dfrac{1}{2}Q,t\right) \simeq T(x,t)- \dfrac{1}{2} Q\cdot \nabla_x T(x,t)$ and $T(r_{2},t)=T\left( r_c+\dfrac{1}{2}Q,t\right)\simeq T(x,t)+ \dfrac{1}{2} Q\cdot \nabla_x T(x,t)$.  Using all these facts, the temperature gradients in \eqref{cd5} can, up to a 1st order approximation, be rewritten as:

\[\nabla_{r_1}\left( T\left( r_1,t\right) \right) = \dfrac{1}{2} \nabla_x (T\ln \Psi)-\dfrac{1}{4} \nabla_x \left[\left(Q\cdot \nabla_x T \right)\ln \Psi  \right]-T\left(\nabla_Q \ln \Psi \right)+\dfrac{1}{2}\nabla_Q \left[\left(Q\cdot \nabla_x T \right)\ln \Psi \right]     \]

and

\[ \nabla_{r_2}\left( T\left( r_2,t\right) \right) = \dfrac{1}{2} \nabla_x (T\ln \Psi)+\dfrac{1}{4} \nabla_x \left[\left(Q\cdot \nabla_x T \right)\ln \Psi  \right]+T\left(\nabla_Q \ln \Psi \right)+\dfrac{1}{2}\nabla_Q \left[\left(Q\cdot \nabla_x T \right)\ln \Psi \right] \]

Therefore, with $\g=\g(x,t) $, \eqref{cd5} implies that

\begin{align}\label{cd6}
 \bigl\lb\stackrel{\cdot}{x} \bigr\rb(x,Q,t) & =  v_0+\g(x,t)\cdot  x  \nonumber\\
& -  \dfrac{k_B}{2\zeta} \left\lbrace \nabla_x \left( T(x,t) \ln \Psi(x,Q,t) \right)+ \nabla_Q \left[\left( Q\cdot \nabla_x T(x,t)\right) \ln \Psi(x,Q,t) \right]   \right\rbrace 
\end{align}


Subtracting over $\nu$ in \eqref{cd4}, and because $F_{1}^{(\phi)}-F_{2}^{(\phi)}=2 F^{(c)}$, $F^{(c)}$ being the connector force,  leads after calculations similar to above to:

\begin{align}\label{cd7}
 \bigl\lb\stackrel{\cdot}{Q} \bigr\rb (x,Q,t)& =  \g(x,t)\cdot  Q 
- \dfrac{k_B}{\zeta} \left\lbrace \dfrac{1}{2}\nabla_x \left[ \left( Q\cdot \nabla_x T(x,t)\right) \ln \Psi(x,Q,t)\right] + T(x,t)\left(\nabla_Q \ln \Psi(x,Q,t) \right) \right\rbrace \nonumber\\
& -  \dfrac{2}{\zeta} F^{(c)}(Q)
\end{align}

The configurational probability density PDE in its general form is (see ~\cite{bird2}):

\beq\label{cd8}
\dfrac{D \Psi}{D t}(x,Q,t)= - \left( \nabla_{x}\cdot  \bigl\lb \stackrel{\cdot}{x} \bigr\rb  \Psi\right)(x,Q,t)  - \left( \nabla_{Q}\cdot  \bigl\lb \stackrel{\cdot}{Q} \bigr\rb \Psi\right)(x,Q,t)
\eeq

In the above, $\dfrac{D}{Dt}=\dfrac{\partial}{\partial t}+ \left(\nabla_x\cdot v\right)\Psi  $ is the material derivative for incompressible fluids. Making use of equations \eqref{cd6}-\eqref{cd7} and upon performing the required calculations leads to

\begin{align}\label{cd9}
\dfrac{D \Psi}{D t} & = \dfrac{k_B}{2\zeta}\nabla_x\cdot \left\lbrace \left[ \nabla_x \left(T \ln \Psi \right)+\nabla_Q \left( \left(Q\cdot \nabla_x T \right)\ln \Psi \right)  \right] \Psi\right\rbrace \nonumber\\
& - \nabla_Q\cdot \left\lbrace \g\cdot Q \Psi -\dfrac{k_B}{\zeta}\left[ \dfrac{1}{2} \nabla_x \left( \left(Q\cdot\nabla_x T \right)\ln \Psi \right) + T \left(\nabla_Q \ln \Psi \right) 
\right] \Psi- \dfrac{2}{\zeta}F^{(c)} \Psi \right\rbrace 
\end{align}

The above equation \eqref{cd9} preserves the normalization condition for the probability density $\Psi$.   

Clearly, the above equation \eqref{cd9} is nonlinear in $\Psi$.  As an aside, its nonlinear pattern contrasts with that of equation (13.17) of ~\cite{bird3} (obtained for a multicomponent/mixture of different polymer chains of Rouse-type with varying temperature and concentration gradients) which is linear in $\Psi$: it is reprinted below for sake of clarity:

\begin{align}\label{cd9-2}
\dfrac{D \Psi}{D t} & = -\sum_{j} \nabla_Q{_j}\cdot \left( \g\cdot Q_j \Psi\right) + \dfrac{k_B }{\zeta}T \sum_{j,k}A_{jk}\nabla_Q{_j}\cdot \left(\nabla_{Q_k} \Psi + \dfrac{1}{k_B T}\nabla_{Q_k} \phi^{(c)} \right) \nonumber\\
& +  \dfrac{k_B }{\zeta}T \sum_{j,k,l}\nabla_{Q_j} \cdot \left(\nabla_x  \ln T  \cdot D_{jkl}Q_l \nabla_{Q_k}\Psi \right) 
\end{align}

where, in the above, $A_{jk}$ and $D_{jkl}$ are the Rouse model's matrices and $\phi^{(c)} $ is the (given) elastic force potential function.  Moreover, as $\Psi$ is considered independent of $x$, \eqref{cd9-2} does not contain any derivative of $\Psi$ with respect to $x$, while our \eqref{cd9} does so.

We are now going to rewrite equation \eqref{cd9} in dimensionless form.   Before proceeding further, we notice from Chapter 2 in ~\cite{mrk} that for a given system of differential equations there are alternative ways of non-dimensionalizing it, depending on the nature of the problem to be studied.   To achieve this goal, dimensionless quantities (identified by starred notation) and relevant dimensionless numbers need first be introduced.  

\begin{enumerate}[$\bullet$]
\item Let $L$ scale the length in the flow direction, $V$ scale velocity, $T_0$ be the reference temperature, $l_0=\sqrt{k_B T_0/H}$ scale the microscopic length scale. $Q_0$ is the maximum spring stretch.  Therefore, $x^{\ast}=\dfrac{x}{L} $, $v^{\ast}=\dfrac{v}{V} $,  $t^{\ast}=\dfrac{t}{L/V} $, $T^{\ast}=\dfrac{T}{T_0} $, $Q^{\ast}=\dfrac{1}{l_0}Q $, $Q_0^{\ast}=\dfrac{Q_0}{l_0}$.  Moreover, $\nabla_{x^{\ast}}=L \nabla_x $ and $\nabla_{Q^{\ast}}=l_0 \nabla_Q $.  

\item The dimensionless Deborah's number $\textbf{De}$ is here taken as $\textbf{De}=\dfrac{\zeta V}{LH}=\dfrac{\zeta V l_0^2}{k_B T_0 L}$, with $H$ the spring constant.  However, because of the interplay between two concomitant different scales (an outer or micro-  and an inner or macro-one),  Deborah's number $\textbf{De}$ may also be introduced as $\textbf{De}=\dfrac{\zeta V}{LH}=\dfrac{\zeta V L}{k_B T_0}$, however with both being essentially the same quantity.  As an aside: if one is tempted to introduce a length scaling factor $s_f=l_0/L$, then the 1st and 3rd term of the equation's r.h.s. will be multiplied by a factor $s^2_f\dfrac{1}{\textbf{De}}$ which is very small, hence they can be neglected, thus profoundly altering the nature of the equation.  Actually, that $s_f=\dfrac{l_0}{L}$ is a small quantity indeed may be inferred from the following: while $L$ is a characteristic macroscopic flow  length, of order of (say) meters, $l_0$ is the  macromolecule stretch (a.k.a. the  end-to-end distance of the unravelled molecule), which is in the sub-micron range, say about $10^{-7}$m, for polymers of industrial importance.  Therefore, as $s_f\sim 10^{-7}\div10^{-6}$m,  can be safely taken as a ``small'' enough  factor and the corresponding terms be dropped.

\end{enumerate}

As it is often common, in order to keep notations simpler, we subsequently drop the $(\cdots)^{\ast} $ notation.  Then, \eqref{cd9} in dimensionless form  becomes:

\begin{align}\label{cd10}
\dfrac{D \Psi}{D t} & =  \dfrac{1}{2 }\dfrac{1}{\textbf{De}}\nabla_x\cdot \left\lbrace \left[ \nabla_x \left(T \ln \Psi \right)+\nabla_Q \left( \left(Q\cdot \nabla_x T \right)\ln \Psi \right)  \right] \Psi\right\rbrace \nonumber\\
& - \nabla_Q\cdot \left( \g\cdot Q \Psi\right) + \dfrac{1}{2 }\dfrac{1}{\textbf{De}} \nabla_Q\cdot  
\left[ \Psi \nabla_x \left( \left(Q\cdot\nabla_x T \right)\ln \Psi \right) \right] \nonumber\\
& +\dfrac{1}{\textbf{De}} \nabla_Q\cdot \left( T \nabla_Q \Psi   \right) +\dfrac{2}{\textbf{De}} \nabla_Q\cdot \dfrac{Q \Psi}{1-\|Q\|^2/Q_0^2} 
\end{align}

\section{Existence Results For The CPD Equation}\label{p2}

\subsection{Introducing The Problem}\label{e}

Let $d\in\{2,3 \}$.  Assume $x\in \Omega\subset\mb{R}^d$, $\Omega_T:=\Omega\times (0,T)$.  Let the ball $B(0,Q_0)\subset \mb{R}^d$, $Q_0>0$, and $Q\in B(0,Q_0)$.  Denote $\widetilde{\Sigma}:=\Omega\times B(0,Q_0)\subset \mb{R}^{2d} $, $\widetilde{\Sigma}_T:=\widetilde{\Sigma}\times (0,T)\subset \mb{R}^{2d+1} $.  Also ${\bf De}>0$.  

Let $v:\Omega_T \mapsto \mb{R}^d$ denote a smooth enough velocity field s.t. $\nabla_x\cdot v=0$ and $v\mid_{\partial \Omega}\cdot \nu =0$, where $\nu$ is the outward normal on $\partial \Omega$.  $\kappa: \Omega_T\mapsto \mathscr{M}_d(\mb{R})$ is the smooth enough velocity gradient, $\theta: \Omega_T\mapsto (0,+\infty)$ a given (known) smooth enough temperature field.    

We search for $f:\widetilde{\Sigma}_T\mapsto [0,+\infty)$, $f=f(x,Q,t)$, solution of (see also \eqref{cd10})

\begin{align}\label{e1}
\dfrac{\partial f }{\partial t}+ v\cdot \nabla_x f & =  \dfrac{1}{2 }\dfrac{1}{\textbf{De}}\nabla_x\cdot \left\lbrace \left[ \nabla_x \left(\theta \ln f \right)+\nabla_Q \left( \left(Q\cdot \nabla_x \theta \right)\ln f \right)  \right] f\right\rbrace \nonumber\\
& - \nabla_Q\cdot \left( \g\cdot Q f\right) + \dfrac{1}{2 }\dfrac{1}{\textbf{De}} \nabla_Q\cdot  
\left[ f \nabla_x \left( \left(Q\cdot\nabla_x \theta \right)\ln f \right) \right] \nonumber\\
& +\dfrac{1}{\textbf{De}} \nabla_Q\cdot \left( \theta \nabla_Q f   \right) +\dfrac{2}{\textbf{De}} \nabla_Q\cdot \dfrac{Q f}{1-\|Q\|^2/Q_0^2} 
\end{align}

complying with the boundary condition 

\begin{equation}\label{e2}
f\mid_{\partial \widetilde{\Sigma}_T\times (0,T)}=0
\end{equation}

and with the initial condition

\begin{equation}\label{e3}
f(t=0)=f_0,\,f_0\,\text{given}
\end{equation}

With the (convenient) change of variable $Q=qQ_0$, $q\in B(0,1)$ and letting $\Sigma=\Omega\times B(0,1)$ and $\Sigma_T=\Sigma\times (0,T)$, then the above introduced problem (we shall stick with the same notations for $f$ and $f_0$) can be restated as following:  investigate the existence of a solution $f:\Sigma_T\mapsto [0,+\infty)$, $f=f(x,q,t)$, to the equation

\begin{align}\label{e4}
\dfrac{\partial f }{\partial t}+ v\cdot \nabla_x f & =  \dfrac{1}{2 }\dfrac{1}{\textbf{De}}\nabla_x\cdot \left\lbrace \left[ \nabla_x \left(\theta \ln f \right)+\nabla_q \left( \left(q\cdot \nabla_x \theta \right)\ln f \right)  \right] f\right\rbrace \nonumber\\
& - \nabla_q\cdot \left( \g\cdot q f\right) + \dfrac{1}{2 }\dfrac{1}{\textbf{De}} \nabla_q\cdot  
\left[ f \nabla_x \left( \left(q\cdot\nabla_x \theta \right)\ln f \right) \right] \nonumber\\
& +\dfrac{1}{Q_0^2}\dfrac{1}{\textbf{De}} \nabla_q\cdot \left( \theta \nabla_q f   \right) +\dfrac{2}{\textbf{De}} \nabla_q\cdot \dfrac{q f}{1-\|q\|^2} 
\end{align}

With the help of the following calculations

\begin{align}\label{e5}
\nabla_x\cdot \left[f\nabla_x\left(\theta \ln f \right)  \right] & =  \nabla_x\cdot  \left( \theta \nabla_x f\right) +\nabla_x\cdot \left[ \left(\nabla_x\theta \right)f\ln f \right]     \nonumber\\
\nabla_x\cdot \left\lbrace  f\nabla_q \left[ \left(q\cdot\nabla_x\theta  \right)\ln f  \right] \right\rbrace & = \nabla_x \cdot \left[ \left(\nabla_x\theta \right)f \ln f \right]+  \nabla_x \cdot\left[ \left(q\cdot \nabla_x\theta \right)\nabla_q f \right] \nonumber\\
\nabla_q\cdot \left\lbrace  f\nabla_x \left[ \left(q\cdot\nabla_x\theta  \right)\ln f  \right] \right\rbrace & = \nabla_q \cdot \left[ \left(\nabla_x^2\theta\cdot q \right)f \ln f \right]+\nabla_q\cdot\left[ \left(q\cdot \nabla_x\theta \right)\nabla_x f \right]  
\end{align}

the problem under scrutiny is restated below:

\begin{align}\label{e6}
\dfrac{\partial f }{\partial t}+ v\cdot \nabla_x f  = & \dfrac{1}{2 }\dfrac{1}{\textbf{De}}\nabla_x\cdot  \left( \theta \nabla_x f\right)+\dfrac{1}{\textbf{De}}\nabla_x\cdot \left[ \left(\nabla_x\theta \right)f\ln f \right]     \nonumber\\  
& +\dfrac{1}{\textbf{De}}\nabla_x \cdot \left[q\cdot \left(\nabla_x\theta \right)\nabla_q f \right] +\dfrac{1}{\textbf{De}}\nabla_q \cdot\left[ \left(q\cdot \nabla_x\theta \right)\nabla_x f \right] \nonumber\\   
& +\dfrac{1}{\textbf{De}}\nabla_q\cdot \left[ \left(\nabla_x^2\theta \cdot q\right)f \ln f \right] + \dfrac{1}{Q_0^2}\dfrac{1}{\textbf{De}} \nabla_q\cdot \left(\theta \nabla_q f \right) \nonumber\\
& - \nabla_q\cdot \left(\kappa \cdot q f \right) + \dfrac{2}{\textbf{De}}\nabla_q\cdot \left(\dfrac{qf}{1-\|q\|^2} \right) 
\end{align}

\beq\label{e7}
f\mid_{\partial \Sigma \times (0,T)}=0
\eeq

\beq\label{e8}
f(t=0)=f_0(x,q),\, 
\eeq

where $f_0:\Sigma\mapsto(0,+\infty)$ is being given.

Assume 

\begin{enumerate}[$\bullet$]
\item $\theta\in L^{\infty}\left(0,T; W^{2,\infty}(\Omega) \right)$, $\theta(x,t)\geq \theta_{\text{min}}>0 $ a.e. $(x,t)\in \Sigma_T $
\item $\g \in L^{\infty}\left(\Omega_T; \mathcal{M}_d\left(\mathbb{R} \right)  \right) $

\item $v\in L^{\infty} \left( 0,T;H^1 \left( \Omega  \right)\right) $, with $\nabla\cdot v=0 $ and $v\cdot \nu =0 $ on $\partial\Omega \times \left(0,T \right)  $, where $\nu$ is the outward normal
\item $f_0\in L^2\left(\Sigma \right) $.

\end{enumerate}

We introduce the continuous function $E:\left[0,+\infty \right) \mapsto \mathbb{R} $, such as

\begin{equation}\label{e9}
E(y)=\begin{cases}
      y\ln y & \text{for $y>0 $  } \\
      0 & \text{for $y=0 $  }
     \end{cases}
\end{equation}

A variational formulation of \eqref{e6}-\eqref{e8} is the following:  find $f\in L^2\left( 0,T; H^1_0\left(\Sigma \right) \right)$, $f\geq0 $, such that for any $\f\in \mathscr{C}^1(0,T) $ with $\f(T)=0 $, and for any  $\psi\in H^1_0\left(\Sigma \right)$ we have

\begin{align}\label{e10}
& -\displaystyle \int_0^T \int_{\Sigma} f \f'(t) \psi(x,q)\ud x\ud q \ud t -\int_{\Sigma} f_0(x,q)\f(0)\psi(x,q)\ud x \ud q \nonumber\\
& +\displaystyle \int_0^T \int_{\Sigma} \left( v\cdot \nabla_x f\right)  \f \psi\ud x\ud q \ud t + \dfrac{1}{ \textbf{De}} \int_0^T \int_{\Sigma} \theta \left( \nabla_x f\cdot \nabla_x \psi\right) \f \ud x\ud q \ud t \nonumber\\
& + \dfrac{1}{ \textbf{De}} \displaystyle \int_0^T \int_{\Sigma} \left( q\cdot \nabla_x \theta\right) \left(\nabla_q f \cdot \nabla_x \psi \right)  \f\ud x\ud q \ud t + \dfrac{1}{ \textbf{De}} \displaystyle \int_0^T \int_{\Sigma} \left( q\cdot \nabla_x \theta\right) \left(\nabla_x f \cdot \nabla_q \psi \right)  \f\ud x\ud q \ud t \nonumber\\
& + \dfrac{1}{Q_0^2 \textbf{De}} \displaystyle \int_0^T \int_{\Sigma} \theta \left( \nabla_q f \cdot \nabla_q  \psi\right) \f \ud x\ud q \ud t + \dfrac{1}{ \textbf{De}} \displaystyle \int_0^T \int_{\Sigma} \left( \nabla_x\theta E(f)\cdot \nabla_x \psi\right)  \f \ud x\ud q \ud t \nonumber\\
& + \dfrac{1}{ \textbf{De}} \displaystyle \int_0^T \int_{\Sigma} \nabla_x^2 \theta \left( q E(f)\cdot \nabla_q \psi \right) \f \ud x\ud q \ud t - \displaystyle \int_0^T \int_{\Sigma}\left(  \g \cdot q f\cdot \nabla_q \psi \right)\f \ud x\ud q \ud t \nonumber\\
& + \dfrac{2}{\textbf{De}} \displaystyle \int_0^T \int_{\Sigma} \left( \dfrac{qf}{1-\|q\|^2}\cdot \nabla_q \psi\right) \f \ud x\ud q \ud t =0
\end{align}

We now introduce a regularization to the aforementioned problem \eqref{e6}-\eqref{e8}.  First, for any small enough $\e>0$, consider the function $g_{\e}:\mb{R}\mapsto \mb{R}$, 

\begin{equation}\label{e11}
g_{\e}(z)=\begin{cases}
           \ln\left(\dfrac{1}{\e} \right) & \text{for $z\geq \dfrac{1}{\e} $  } \\ 
           \ln z & \text{for $\e \leq z \leq \dfrac{1}{\e} $  } \\ 
           \ln \e & \text{for $\dfrac{1}{\ln \e}\leq z \leq \e $  } \\
           \dfrac{1}{z} & \text{for $z\leq \dfrac{1}{\ln \e} $  } \\
          \end{cases}
\end{equation}

Denote $E_\e:\mb{R}\mapsto\mb{R}$, $E_\e(z)=z g_\e(z)$.  The announced regularized problem reads: find $f_\e: \Sigma\mapsto\mb{R}$ that solves

\begin{align}\label{e12}
& \dfrac{\partial f_\e}{\partial t}+v\cdot \nabla_x f_\e = \dfrac{1}{ \textbf{De}} \nabla_x\cdot \left(\theta \nabla_x f_\e \right) +  \dfrac{1}{ \textbf{De}} \nabla_x\cdot \left[ \left(q\cdot \nabla_x \theta \right)\nabla_q f_\e \right]  \nonumber\\ 
& + \dfrac{1}{ \textbf{De}} \nabla_q\cdot \left[ \left(q\cdot \nabla_x \theta \right)\nabla_x f_\e \right] + \dfrac{1}{Q_0^2 \textbf{De}}  \nabla_q\cdot \left(\theta \nabla_q f_\e\right) + \dfrac{1}{ \textbf{De}}\nabla_x \cdot \left[ \left( \nabla_x \theta \right)E_\e(f_\e) \right] \nonumber\\
& + \dfrac{1}{ \textbf{De}} \nabla_q\cdot \left[ \left(\nabla_x^2 \theta  \right)qE_\e(f_\e)  \right] - \nabla_q\cdot \left(\g \cdot q f_\e \right) + \dfrac{2}{ \textbf{De}}\nabla_q \cdot \left(\dfrac{f_e q}{\e+1-\|q \|^2} \right) 
\end{align}

\beq\label{e13}
f_\e=0,\,\text{on}\, \partial\Sigma\times (0,T)
\eeq

\beq\label{e14}
f_\e(t=0)=f_0 \,\text{on}\, \Sigma.
\eeq

Of notice: the variational formulation in \eqref{e12} is the same as the one in \eqref{e10} whereupon replacing $f$ by $f_\e$, $E(f)$ by $E_\e(f_\e) $ and $\dfrac{1}{1-\|q\|^2} $ by $\dfrac{1}{\e+1-\|q\|^2} $.

\subsection{Proof Of The Existence Of A Solution Of The Regularized Problem For Any Small Enough $\e>0$.}\label{exreg}

We shall make use of a fixed point theorem method.

Let the operator $S_\e:L^2\left(\Sigma_T \right)\mapsto  L^2\left(\Sigma_T \right) $, $S_\e\left(\widetilde{f} \right)=f_\e$, where $f_\e$ solves the following linear problem:

\begin{align}\label{r1}
& \dfrac{\partial f_\e}{\partial t}+v\cdot \nabla_x f_\e = \dfrac{1}{ \textbf{De}} \nabla_x\cdot \left(\theta \nabla_x f_\e \right) +  \dfrac{1}{ \textbf{De}} \nabla_x\cdot \left[ \left(q\cdot \nabla_x \theta \right)\nabla_q f_\e \right]  \nonumber\\ 
& + \dfrac{1}{ \textbf{De}} \nabla_q\cdot \left[ \left(q\cdot \nabla_x \theta \right)\nabla_x f_\e \right] + \dfrac{1}{Q_0^2}\dfrac{1}{ \textbf{De}}  \nabla_q\cdot \left(\theta \nabla_q f_\e\right) + \dfrac{1}{  \textbf{De}}\nabla_x \cdot \left[ \left( \nabla_x \theta \right)f_\e g_\e\left(\widetilde{f} \right)  \right] \nonumber\\
& + \dfrac{1}{ \textbf{De}} \nabla_q\cdot \left[ \left(\nabla_x^2 \theta  \right)q f_\e g_\e\left(\widetilde{f} \right)  \right] - \nabla_q\cdot \left(\g \cdot q f_\e \right) + \dfrac{2}{ \textbf{De}}\nabla_q \cdot \left(\dfrac{f_e q}{\e+1-\|q \|^2} \right) 
\end{align}

with 

\beq\label{r2}
f_\e=0,\,\text{on}\, \partial\Sigma\times (0,T)
\eeq

\beq\label{r3}
f_\e(t=0)=f_0 \,\text{on}\, \Sigma.
\eeq

In order to avoid cumbersome notations, we denote by $f_\e$ the solution of \eqref{r1} as well as that of \eqref{e12}.

\subsubsection{Proof Of The Existente And Uniqueness Of A Variational Solution To Equations \eqref{r1}-\eqref{r3}.}\label{vreg}

First we take on to obtaining the variational formulation that corresponds to \eqref{r1}-\eqref{r3}.  To achieve this, we first introduce the application $a_\e:(0,T)\times L^2(\Sigma)\times H^1_0(\Sigma)\times H^1_0(\Sigma)\mapsto \mb{R}$,

\begin{align}\label{r4}
& a_\e\left(t,r,\psi,\xi \right) = \displaystyle \int_{\Sigma} v\cdot \nabla_x \psi \xi + \dfrac{1}{ \textbf{De}} \displaystyle \int_{\Sigma} \theta \nabla_x \psi\cdot \nabla_x \xi + \dfrac{1}{ \textbf{De}} \displaystyle \int_{\Sigma} \left( q\cdot \nabla_x \theta\right) \left(\nabla_q \psi \cdot \nabla_x \xi\right) \nonumber\\
& + \dfrac{1}{ \textbf{De}} \displaystyle \int_{\Sigma} \left( q\cdot \nabla_x \theta\right) \left(\nabla_x \psi \cdot \nabla_q \xi\right) + \dfrac{1}{Q_0^2 \textbf{De}} \displaystyle \int_{\Sigma} \theta \nabla_q \psi \cdot \nabla_q \xi + \dfrac{1}{ \textbf{De}} \displaystyle \int_{\Sigma} \left(\nabla_x \theta\right)  \psi g_\e(r)\cdot  \nabla_x \xi \nonumber\\
& + \dfrac{1}{ \textbf{De}} \displaystyle \int_{\Sigma} \left( \nabla_x^2 \theta\right)  q \psi g_\e(r)\cdot \nabla_x \xi - \displaystyle \int_{\Sigma} \g \cdot q \psi \cdot \nabla_q \xi + \dfrac{2}{\textbf{De}} \displaystyle \int_{\Sigma} \dfrac{\psi q}{\e+1-\|q\|^2}\cdot \nabla_q \xi
\end{align}

for any $t\in (0,T)$, $r\in L^2(\Sigma)$, $\psi,\, \xi\in H^1_0(\Sigma)$. 

We now use the fact that $\theta\in W^{2,\infty}(\Omega)$, $\g \in L^\infty(\Omega)$, $g_\e(r)\in L^\infty(\Sigma)$, and the function $ q\mapsto \dfrac{1}{\e+1-\|q\|^2}\in L^\infty\left(  B(0,1)\right)  $.  By making use of the Cauchy-Schwarz inequality on all terms making-up $a_\e$ it is easily deduced that 

\begin{equation}\label{r4-2}
\left| a_\e\left(t,\widetilde{f},\psi,\xi \right)  \right|\leq c \|\psi\|_{H^1_0(\Sigma)}\|\xi\|_{H^1_0(\Sigma)},\, \forall (t,\psi,\xi)\in (0,T)\times \left(H^1_0(\Sigma) \right)^2 
\end{equation}

where $c$ is a constant independent of $t$ (but dependent on $\e$).

Let $b:(0,T)\times H^1_0(\Sigma) \times H^1_0(\Sigma)\mapsto \mb{R}$ be such that:

\begin{align*}
&  b\left(t,\psi,\xi \right) = \dfrac{1}{ \textbf{De}} \displaystyle \int_{\Sigma} \theta \nabla_x \psi\cdot \nabla_x \xi + \dfrac{1}{ \textbf{De}} \displaystyle \int_{\Sigma} \left( q\cdot \nabla_x \theta\right) \left(\nabla_q \psi\cdot \nabla_x \xi\right) \\
& + \dfrac{1}{ \textbf{De}} \displaystyle \int_{\Sigma} \left( q\cdot \nabla_x \theta\right) \left(\nabla_x \psi \cdot \nabla_q \xi\right) + \dfrac{1}{Q_0^2 \textbf{De}} \displaystyle \int_{\Sigma} \theta \nabla_q \psi \cdot \nabla_q \xi,\, \forall  (t,\psi,\xi)\in (0,T)\times \left(H^1_0(\Sigma) \right)^2 
\end{align*}

We have:

\begin{lemma}\label{rl1}
There exist $c_1>0$, $c_2\in\mb{R}$ (not independent of $\e$) such that 

\[ a_\e\left(t,\widetilde{f},\psi,\psi \right) \geq c_1  \|\psi\|^2_{H^1_0(\Sigma)}-c_2 \|\psi\|^2_{L^2(\Sigma)},\, \forall \psi\in H^1_0(\Sigma),\, \forall t\in(0,T). \]

\end{lemma}

\begin{proof}

Notice now that 

\begin{align*}
&  b\left(t,\psi,\psi \right) = \dfrac{1}{ \textbf{De}} \displaystyle \int_{\Sigma} \theta \left\|\nabla_x \psi\right\|^2 + \dfrac{2}{ \textbf{De}} \displaystyle \int_{\Sigma} \left( q\cdot \nabla_x \theta\right) \left(\nabla_x \psi\cdot \nabla_q \psi\right) + \dfrac{1}{Q_0^2 \textbf{De}} \displaystyle \int_{\Sigma} \theta \left\|\nabla_q \psi \right\|^2
\end{align*}
 
Next, by Cauchy-Schwarz one gets

\[ \left| \displaystyle \int_{\Sigma} \left( q\cdot \nabla_x \theta\right) \left(\nabla_x \psi\cdot \nabla_q \psi\right) \right| \leq \|\theta\|_{W^{1,\infty}(\Omega)} \left\|\nabla_x \psi \right\|_{L^2(\Sigma)} \left\|\nabla_q \psi \right\|_{L^2(\Sigma)} \]

from which one obtains

\begin{align*}
b\left(t,\psi,\psi \right) & \geq  \dfrac{\theta_{\text{min}}}{ \textbf{De}}\left\|\nabla_x \psi \right\|^2_{L^2(\Sigma)} + \dfrac{\theta_{\text{min}}}{Q_0^2 \textbf{De}}\left\|\nabla_q \psi \right\|^2_{L^2(\Sigma)} \nonumber\\
& - \dfrac{2}{ \textbf{De}}\left\|\nabla_x \theta \right\|_{L^\infty(\Omega)}\left\|\nabla_x \psi \right\|_{L^2(\Sigma)} \left\|\nabla_q \psi \right\|_{L^2(\Sigma)}
\end{align*}

Let the symmetric matrix 

\begin{eqnarray*}
A =
\left( \begin{array}{cc}
       \dfrac{\theta_{\text{min}}}{ \textbf{De}} & -\dfrac{1}{ \textbf{De}} \|\nabla_x \theta\|_{L^\infty(\Omega)}\\
       -\dfrac{1}{ \textbf{De}} \|\nabla_x \theta\|_{L^\infty(\Omega)} & \dfrac{\theta_{\text{min}}}{Q_0^2 \textbf{De}}
       \end{array}  \right)
\end{eqnarray*}

with the help of which the above inequality can be re-written as

\[b(t,\psi,\psi)\geq \left\langle Az,z \right\rangle  \]
 
with  
 
\begin{eqnarray*}
\mb{R}^2\ni z =
\left( \begin{array}{c}
        \|\nabla_x \psi\|_{L^2(\Sigma)} \\
        \|\nabla_q \psi\|_{L^2(\Sigma)} 
       \end{array}  \right)
\end{eqnarray*}

Clearly $A$ is a symmetric and positive definite matrix iff $\det{A}>0 $, i.e.

\[\dfrac{\theta^2_{\text{min}}}{ Q_0^2 \textbf{De}^2}> \dfrac{1}{ \textbf{De}^2}\|\nabla_x \theta\|^2_{L^\infty(\Omega)}.  \]
 
Because of this fact we shall consider the following assumption on the problem data:

\begin{equation}\label{r5}
  \theta^2_{\text{min}}> Q_0^2 \|\nabla_x \theta\|^2_{L^\infty(\Omega)}
\end{equation}

Invoking the above assumption leads to the existence of a $\widetilde{\lambda}_m>0$ depending on the problem data and such that

\[ b\left(t,\psi,\psi \right) \geq  \widetilde{\lambda}_m \left\|\nabla_{x,q}\psi\right\|^2_{L^2(\Sigma)},\, \forall \psi\in H^1_0(\Sigma)  \]

Using now Poincar\'e's inequality one deduces the existence of a $\lambda_m>0$ such that 

\begin{equation}\label{r6}
b\left(t,\psi,\psi \right) \geq \lambda_m \left\|\psi\right\|^2_{H^1_0(\Sigma)}, \, \forall \psi\in H^1_0(\Sigma) 
\end{equation}

We now proceed to upper bound the terms in $a_\e (t,\psi,\psi) $ that do not appear in the definition of $b\left(t,\psi,\psi \right)$.  Upon using the Cauchy-Schwarz's inequality it is easily seen there exists $c_\e$ such that

\begin{align*}
& \bigg|  \displaystyle \int_{\Sigma} \left( v\cdot \nabla_x \psi\right)  \psi + \dfrac{1}{ \textbf{De}} \displaystyle \int_{\Sigma} \left(\nabla_x \theta\right)  \psi g_\e(\widetilde{f})\cdot  \nabla_x \psi + \dfrac{1}{ \textbf{De}} \displaystyle \int_{\Sigma} \left( \nabla_x^2 \theta\right)   \psi g_\e(\widetilde{f})q\cdot \nabla_x \psi \nonumber\\
& - \displaystyle \int_{\Sigma} \g \cdot q \psi \cdot \nabla_q \psi + \dfrac{2}{\textbf{De}} \displaystyle \int_{\Sigma} \dfrac{\psi q}{\e+1-\|q\|^2}\cdot \nabla_q \psi  \bigg| \nonumber\\
& \leq c_\e \|\psi\|_{L^2(\Sigma)} \|\psi\|_{H^1(\Sigma)},\,\forall (t,\psi)\in(0,T)\times H^1_0(\Sigma)
\end{align*}

Next on, for any $\eta>0$, by the Young's inequality we obtain

\[ c_\e \|\psi\|_{L^2(\Sigma)} \|\psi\|_{H^1(\Sigma)}\leq \eta \|\psi\|^2_{H^1(\Sigma)} + \dfrac{1}{4\eta}c^2_\e \|\psi\|^2_{L^2(\Sigma)}   \]

We then deduce that

\begin{align*}
& \displaystyle \int_{\Sigma} \left( v\cdot \nabla_x \psi\right)  \psi + \dfrac{1}{ \textbf{De}} \displaystyle \int_{\Sigma} \left(\nabla_x \theta\right)  \psi g_\e(\widetilde{f})\cdot  \nabla_x \psi + \dfrac{1}{ \textbf{De}} \displaystyle \int_{\Sigma} \left( \nabla_x^2 \theta\right)   \psi g_\e(\widetilde{f})q\cdot \nabla_x \psi \nonumber\\
& - \displaystyle \int_{\Sigma} \g \cdot q \psi \cdot \nabla_q \psi + \dfrac{2}{\textbf{De}} \displaystyle \int_{\Sigma} \dfrac{\psi q}{\e+1-\|q\|^2}\cdot \nabla_q \psi   \nonumber\\
& \geq -\eta \|\psi\|^2_{H^1(\Sigma)} - \dfrac{1}{4\eta}c^2_\e \|\psi\|^2_{L^2(\Sigma)}
\end{align*}

Taking now $\eta=\dfrac{\lambda_m}{2}$ and using \eqref{r6} we obtain the result taking $c_1=\dfrac{\lambda_m}{2}$ and $c_2=\dfrac{c_\e^2}{2\lambda_m}$.
\end{proof}

A variational formulation of \eqref{r1}-\eqref{r3} is the following:  find $f_\e\in L^2 
\left(0,T; H^1_0(\Sigma)\right) \cap L^\infty\left(0,T; L^2(\Sigma) \right)   $ solution to 

\beq\label{r7}
\dfrac{d}{dt}\left(f_\e,\psi \right)_{L^2(\Sigma)} + a_\e\left(t,\widetilde{f},f_\e,\psi \right)=0,\, \forall \psi\in H^1_0(\Sigma)
\eeq

with

\beq\label{r8}
f_\e(t=0)=f_0
\eeq

Remark that \eqref{r7} is to be understood in the sense of distributions.  For any $\psi \in H^1_0(\Sigma)$, and for any $\f\in \ms{C}^1(0,T)$, $\f(T)=0$,   using \eqref{r8} one has:

\begin{align}\label{r9}
& -\left(f_0,\psi \right)_{L^2(\Sigma)} \f(0)-\displaystyle\int_0^T \left(f_\e,\psi \right)_{L^2(\Sigma)} \f(t)\ud t +  \displaystyle\int_0^T  a_\e\left(t,\widetilde{f},f_\e,\psi \right)\f(t) \ud t=0
\end{align}

Theorem 4.1 on page 257 together with Remark 4.3 on page 258 of ~\cite{limg} grants the existence of a unique solution to \eqref{r9} (due to \eqref{r4-2} and Lemma \ref{rl1}).

Remark that we can introduce the function $A_\e:(0,T)\times L^2(\Sigma)\times H^1_0(\Sigma)\mapsto H^{-1}(\Sigma)$, $A_\e=A_\e\left(t,\widetilde{f},\psi \right) $, in the following way:  $\displaystyle H^1_0(\Sigma)\ni \xi \xrightarrow[\, ]{A_\e}  a_\e\left(t,\widetilde{f},\psi,\xi \right)\in \mb{R} $.  The mapping $A_\e$ is also linear and continuous from $H^1_0(\Sigma)$ to  $H^{-1}(\Sigma)$ due to \eqref{r4-2}.  Then the mapping $t\mapsto A_\e\left(t,\widetilde{f}(t),f_\e(t) \right)$ is an element of $L^2\left(0,T; H^{-1}(\Sigma) \right)$ because $f_\e\in L^2\left(0,T; H^1_0(\Sigma) \right) $.  Then equation \eqref{r7} can be re-written as 

\beq\label{r10}
\dfrac{d}{dt}f_\e+A_\e\left(t,\widetilde{f},f_\e \right)=0
\eeq

hence $\displaystyle \dfrac{d}{dt}f_\e \in  L^2\left(0,T; H^{-1}(\Sigma) \right)$.

Therefore the solution $f_\e$ is such that $f_\e\in X_T$, where 

\[ X_T:=\left\{f\in L^2\left(0,T; H^1_0(\Sigma) \right):\, \dfrac{df}{dt}\in  L^2\left(0,T; H^{-1}(\Sigma) \right)  \right\} \]

is a Banach space endowed with the norm 

\[ \|f\|_{X_T}=\|f\|_{L^2\left(0,T; H^1_0(\Sigma) \right)} + \left\|\dfrac{df}{dt} \right\|_{L^2\left(0,T; H^{-1}(\Sigma) \right) }  \]

Moreover, \eqref{r8} is meaningful because of the continuous embedding $X_T\subset \ms{C}\left((0,T); L^2(\Sigma) \right)  $.  Therefore the mapping $S_\e$ is well defined.   Remember that we also have the compact embedding $X_T \subset L^2\left(\Sigma_T \right) $.

\subsubsection{Estimates For The Solution To The Problem \eqref{r1}-\eqref{r3}.}\label{es}

Since $\displaystyle f_\e \in  L^2\left(0,T; H^1_0(\Sigma) \right)$ and because $L^2\left(0,T; H^{-1}(\Sigma) \right)$ is the dual space of $L^2\left(0,T; H^1_0(\Sigma) \right)$, we apply \eqref{r10} to $f_\e$.  One has

\[\left\langle \dfrac{d f_\e}{dt}, f_\e \right\rangle \stackrel{\ms{D}'(0,T)}{=} \dfrac{1}{2} \dfrac{d}{dt}\left(\left\|f_\e \right\|^2_{L^2(\Sigma)} \right)    \]

Therefore,

\[\dfrac{1}{2} \dfrac{d}{dt}\left(\left\|f_\e \right\|^2_{L^2(\Sigma)} \right)  + a_\e\left( t,f_\e,f_\e\right)=0   \]

Using the result stated in Lemma \ref{rl1} gives

\beq\label{r11}
\dfrac{d}{dt}\left(\left\|f_\e \right\|^2_{L^2(\Sigma)} \right)+2 c_1 \left\|f_\e \right\|^2_{H^1(\Sigma)}   \leq 2 c_2 \left\|f_\e \right\|^2_{L^2(\Sigma)}
\eeq

Further use of Gronwall's inequality on \eqref{r11} entails

\[ \left\|f_\e \right\|^2_{L^2(\Sigma)}\leq  \left\|f_0 \right\|^2_{L^2(\Sigma)} e^{2 c_2 T},\, \forall t\in(0,T) \]

Next, integrating \eqref{r11} w.r.t. $t\in(0,T)$ gives

\[\left\|f_\e \right\|_{L^2\left(0,T; H^1(\Sigma) \right) }\leq \sqrt{\dfrac{T}{2c_1}}\left\|f_0 \right\|_{L^2(\Sigma)} e^{ c_2 T}  \]

and with the help of \eqref{r10} one gets

\[\left\|\dfrac{d f_\e}{dt} \right\|_{ L^2\left(0,T; H^{-1}(\Sigma) \right) } \leq c_3.  \]

We then deduce the existence of a constant  $c=c(\e)$ s.t.:

\beq\label{r12}
\left\|f_\e \right\|_{X_T}\leq c
\eeq

\subsubsection{Proof Of The Fixed Point Result}\label{fp}

Schauder's fixed-point Theorem is used  to proving the existence of at least one variational solution to the problem \eqref{e12}-\eqref{e14}.  From \eqref{r12} we have that $S_\e\left(L^2(\Sigma_T) \right)  $ is relatively compact in $L^2(\Sigma_T)$.  All constitutive requirements of Schauder's fixed-point Theorem are met save for the continuity of $S_\e$, fact we shall ascertain in the following.

\begin{lemma}\label{fp1}
$S_\e$ is a continuous mapping from  $L^2(\Sigma)$ to $L^2(\Sigma)$.
\end{lemma}

\begin{proof}
Let $\widetilde{f}\in L^2(\Sigma_T) $ be a fixed element, and consider a converging sequence 

\[ \displaystyle L^2(\Sigma_T)\ni \widetilde{f}_k \xrightarrow[k\to\infty]{L^2(\Sigma_T)}\widetilde{f} \]

Denote $f_{\e,k}=S_\e\left( f_k\right)$ and $f_{\e}=S_\e\left( \widetilde{f}\right)  $.  We need to prove that 

\[ \displaystyle  f_{\e,k} \xrightarrow[k\to\infty]{L^2(\Sigma)} f_\e \]

We have, as in \eqref{r12}, that $\displaystyle \left\|f_{\e,k} \right\|_{X_T}\leq c $, where $c$ may depend on $\e$ but not on $k$.  From the property of compactness we infer there exist $ \widehat{f}_\e\in X_T$, and a subsequence (also denoted by) $f_{\e,k}$ s.t.

\[f_{\e,k} \xrightharpoonup[k\to\infty]{L^2\left(0,T; H^1_0(\Sigma) \right)} \widehat{f}_\e,\,\text{weakly} \]

\[\dfrac{d f_{\e,k}}{dt} \xrightharpoonup[k\to\infty]{L^2\left(0,T; H^{-1}(\Sigma) \right)} \dfrac{d \widehat{f}_\e}{dt},\,\text{weakly} \]

\[f_{\e,k} \xrightarrow[k\to\infty]{L^2\left( \Sigma_T \right)} \widehat{f}_\e,\,\text{strongly} \]

Therefore $f_{\e,k}$ satisfies (see also \eqref{r9}), for any $\psi \in H^1_0(\Sigma)$, for any $\f\in \ms{C}^1(0,T)$ s.t.  $\f(T)=0$,

\begin{align}\label{fp2}
& -\left(f_0,\psi \right)_{L^2(\Sigma)} \f(0) -\displaystyle\int_0^T \left(f_{\e,k},\psi \right)_{L^2(\Sigma)} \f'(t)\ud t +  \displaystyle\int_0^T  a_{\e,k}\left(t,\widetilde{f}_k,f_{\e,k},\psi \right)\f(t) \ud t=0
\end{align}

We now prove that passing to the limit in \eqref{fp2}, for $k\to\infty$, leads to obtaining \eqref{r9} with $f_\e$ being replaced by $\widehat{f}_\e$.  All the limit related calculations are obvious due to the established weak convergence $\displaystyle f_{\e,k} \xrightharpoonup[k\to\infty]{L^2\left(0,T; H^1_0(\Sigma) \right)} \widehat{f}_\e $, excepting the following convergences:

\begin{align}\label{fp3}
& \displaystyle \int_0^T \int_\Sigma  \nabla_x \theta f_{\e,k} g_{\e}\left(\widetilde{f}_k \right)\cdot \nabla_x \xi \f(t)\xrightarrow[k\to\infty]{\,}  \displaystyle \int_0^T \int_\Sigma \nabla_x \theta \widehat{f}_{\e} g_{\e}\left(\widetilde{f} \right)\cdot \nabla_x \xi \f(t)
\end{align}

\begin{align}\label{fp4}
& \displaystyle \int_0^T \int_\Sigma  \nabla^2_x \theta q f_{\e,k} g_{\e}\left(\widetilde{f}_k \right)\cdot \nabla_x \xi \f(t) \xrightarrow[k\to\infty]{\,}  \displaystyle \int_0^T \int_\Sigma \nabla^2_x \theta q \widehat{f}_{\e} g_{\e}\left(\widetilde{f} \right)\cdot \nabla_x \xi \f(t)
\end{align}

The above convergences hold true in wake of the strong convergence $\displaystyle g_{\e}\left(\widetilde{f}_k \right)\xrightarrow[k\to\infty]{L^2(\Sigma_T)}g_{\e}\left(\widetilde{f} \right)$, which is manifest in view of the fact that the function $z\mapsto g_\e(z)$ is an element of $W^{1,\infty}(\mb{R})$.

We eventually obtain the desired limit problem being satisfied by $\widehat{f}_{\e} $.  Moreover, the uniqueness of $\widehat{f}_{\e} $ tells that all sequences $\left\{f_{\e,k} \right\}_{k\in\mb{N}} $ converge towards $\widehat{f}_{\e} $, fact that ends the proof.
\end{proof}

Therefore, by Schauder's fixed-point Theorem we have a variational solution to the regularized problem \eqref{e12}-\eqref{e14}.

\subsection{Estimates Uniform In $\e$.}\label{ee}

We draw some inspiration from ~\cite{krz} and from ~\cite{chp1} for the obtention of $L^1$ - type estimates.  Here we obtain $\e$-free estimates for the solution $f_\e$ that solves \eqref{e12}-\eqref{e14}, in order to calculate the limit for $\e\to0$.  Moreover, the solution $f_\e$ is in fact a variational solution for it solves (see also \eqref{r9} and \eqref{e10})

\begin{align}\label{ee1}
& -\left(f_0,\psi \right)_{L^2(\Sigma)} \f(0)-\displaystyle\int_0^T \left(f_\e,\psi \right)_{L^2(\Sigma)} \f(t)\ud t +  \displaystyle\int_0^T  a_\e\left(t,f_\e,f_\e,\psi \right)\f(t) \ud t=0
\end{align}

for any $\psi \in H^1_0(\Sigma)$, and for any $\f\in \ms{C}^1(0,T)$ s.t. $\f(T)=0$.  This means that $f_\e\in X_T$ solves

\beq\label{ee2}
\dfrac{d f_\e}{dt}+A_\e\left(t,f_\e,f_\e \right)=0
\eeq

\subsubsection{$L^1(\Sigma)$  Estimates.}\label{eel}
 
For any $\eta>0$ we introduce the functions:

\begin{enumerate}[$\bullet$]
 \item an approximation of the function $\text{sgn}(y)$
 \[\beta_\eta:\mb{R}\mapsto\mb{R},\, \beta_\eta(y):=\dfrac{y}{\sqrt{y^2+\eta}}  \]
 \item an approximation of function $|y|$
 \[\gamma_\eta :\mb{R}\mapsto\mb{R},\,\gamma_\eta(y)=\sqrt{y^2+\eta}  \]
\end{enumerate}

Remark that $\beta_\eta,\gamma_\eta\in\ms{C}^\infty(\mb{R})$ and that 

\[\gamma'_\eta=\beta_\eta,\, \beta'_\eta(y)=\dfrac{\eta}{\left(y^2+\eta \right)^{3/2} }>0,\,\forall y\in\mb{R}  \]

Making use of \eqref{ee2} upon $\beta_\eta\left(f_\e \right)  $ gives

\beq\label{ee3}
\left\langle \dfrac{d f_\e}{dt},\beta_\eta\left(f_\e \right) \right\rangle + a_\e \left(t,f_\e,f_\e,\beta_\eta\left(f_\e \right) \right) = 0
\eeq
 
Also remark that as $\beta_\eta\in W^{1,\infty}(\mb{R}),\, \beta_\eta(0)=0 $, it implies $\beta_\eta\left(f_\e \right)\in L^2\left(0,T; H^1_0(\Sigma)\right)  $.  
 
Since $\gamma'_\eta\left(f_\e \right) =\beta_\eta\left(f_\e \right) $, then

\[\left\langle \dfrac{d f_\e}{dt},\beta_\eta\left(f_\e \right) \right\rangle = \dfrac{d}{dt} \displaystyle \int_\Sigma \gamma_\eta\left(f_\e \right)\ud x \ud q \]
 
We now have

\begin{align*}
& b_\e\left(t,f_\e,\beta_\eta\left(f_\e \right) \right) = \dfrac{1}{ \textbf{De}} \displaystyle \int_\Sigma \beta'_\eta\left(f_\e \right)  \theta \left\|\nabla_x f_\e \right\|^2 + \dfrac{2}{ \textbf{De}} \displaystyle \int_\Sigma \beta'_\eta\left(f_\e \right) \left( q\cdot \nabla_x\theta\right)  \left( \nabla_q f_\e \cdot  \nabla_x f_\e\right) \\
& + \dfrac{1}{Q_0^2 \textbf{De}} \displaystyle \int_\Sigma \beta'_\eta\left(f_\e \right)  \theta \left\|\nabla_q f_\e \right\|^2 \\
& \geq \displaystyle \int_\Sigma \beta'_\eta \left(f_\e \right) \bigg[ \dfrac{1}{ \textbf{De}} \theta_{\text{min}}\left\|\nabla_x f_\e \right\|^2 - \dfrac{2}{ \textbf{De}}\left\|\nabla_x \theta \right\|_{L^\infty(\Omega)}\left\|\nabla_x f_\e \right\|\left\|\nabla_q f_\e \right\| \\
& +\dfrac{1}{Q_0^2 \textbf{De}}\theta_{\text{min}} \left\|\nabla_q f_\e \right\|^2  \bigg]
\end{align*}

Due to the assumption \eqref{r5} we have

\[ b_\e\left(t,f_\e,\beta_\eta\left(f_\e \right) \right)\geq 0 \]

Remark that, owing to the assumption made on $v$, one has 

\begin{align*}
& \displaystyle \int_\Sigma v\cdot \nabla_x f_\e \beta_\eta \left(f_\e \right)  = \displaystyle \int_\Sigma v\cdot \nabla_x  \gamma_\eta \left(f_\e \right) = - \displaystyle \int_\Sigma \nabla_x\cdot v \gamma_\eta \left(f_\e \right) \\
& + \displaystyle \int_{\partial \Sigma} \left( v\cdot\nu\right)  \gamma_\eta \left(f_\e \right) =0
\end{align*}

Then, from \eqref{ee3} we get

\begin{align*}
& \dfrac{d}{dt} \displaystyle \int_\Sigma \gamma_\eta \left(f_\e \right) \ud x \ud q + \dfrac{1}{ \textbf{De}} \displaystyle \int_\Sigma f_\e g_\eta \left(f_\e \right)\cdot \beta'_\eta \left(f_\e \right) \nabla_x f_\e \cdot \left(2 \nabla_x \theta + q \nabla_x^2\theta  \right) \\
& - \displaystyle \int_\Sigma \g \cdot q f_\e \cdot \beta'_\eta \left(f_\e \right) \nabla_x f_\e  
+ \dfrac{1}{\textbf{De}} \displaystyle \int_\Sigma \dfrac{f_\e q}{\e+1-\|q\|^2}\cdot \beta'_\eta \left(f_\e \right)\nabla_q f_\e \leq 0
\end{align*}

We now integrate w.r.t. $t$ from 0 to a arbitrarily fixed $t\in (0,T)$ and take the limit $\eta\to 0$.  We shall make repeated use of Lemma 3.2 of ~\cite{chp1} with $h=f_\e$.  We deduce that

\begin{align*}
& \displaystyle \lim_{\eta\to0}\bigg\{ \displaystyle \int_{\Sigma_T}f_\e \beta'_\eta \left(f_\e \right)\big[ g_\eta \left(f_\e \right)\nabla_x f_\e \cdot  \left(2 \nabla_x \theta + q\nabla_x^2 \theta \right)\dfrac{1}{ \textbf{De}}    \\
&- \g \cdot q \cdot \nabla_x f_\e +  \dfrac{2}{ \textbf{De}} \dfrac{ q}{\e+1-\|q\|^2}\cdot \nabla_q f_\e \big]\bigg\}=0
\end{align*}

We also have that

\[ \displaystyle \lim_{\eta\to0} \displaystyle \int_\Sigma \gamma_\eta \left(f_\e \right)\ud x \ud q = \displaystyle \int_\Sigma \left|f_\e  \right| \ud x \ud q,\, \text{a.e.} t\in(0,T) \]

By the use of Lebesgue's dominated convergence Theorem we have that

\beq\label{ee4}
\left\| f_\e\right\|_{L^\infty\left(0,T; L^1(\Sigma) \right) }\leq \left\| f_0 \right\|_{L^1(\Sigma)}
\eeq

\%
\subsubsection{Estimates Uniform In $\e$ in Functional Space $X_T$.}\label{ex}

Apply \eqref{ee2} to $f_\e$ to get

\beq\label{ee5}
\dfrac{1}{2}\dfrac{d}{dt} \displaystyle \int_\Sigma f_\e^2 + a_\e\left(t,f_\e, f_\e, f_\e \right) =0
\eeq

Due to the assumption on $v$ one has 

\begin{align}\label{ee6}
& \displaystyle \int_\Sigma \left( v\cdot \nabla_x f_\e\right)  f_\e = \dfrac{1}{2} \displaystyle \int_\Sigma v\cdot \nabla_x \left( f_\e\right)^2 = -\dfrac{1}{2} \displaystyle \int_\Sigma \left(\nabla_x\cdot v \right)f^2_\e + \dfrac{1}{2} \displaystyle \int_{\partial\Sigma} \left( v\cdot \nabla_x\right) f^2_\e =0 
\end{align}

Due to Hardy's inequality,

\[ \left\| \dfrac{\psi}{1-\|q\|^2} \right\|_{L^2 \left( B(0,1) \right) } \leq c_H \left\|\nabla_q \psi  \right\|_{L^2\left(B(0,1)\right)},\, \forall \psi \in H^1_0 \left( B(0,1) \right) \]

with $c_H>0$ being Hardy's constant. Next, integrating on $\Omega$ gives

\beq\label{ee7}
\left\| \dfrac{\psi}{1-\|q\|^2} \right\|_{L^2 \left( \Sigma \right) } \leq c_H \left\|\nabla_q \psi  \right\|_{L^2\left(\Sigma\right)},\, \forall \psi \in H^1_0 \left( B(0,1) \right)
\eeq

Then we have

\begin{align*}
 & \left|\int_\Sigma \dfrac{\psi q}{\e+1-\|q\|^2}\cdot \nabla_q \psi  \right|\leq \int_\Sigma \dfrac{1}{1-\|q\|^2}|\psi|\left\|\nabla_q \psi  \right\| \leq \left\| \dfrac{\psi}{1-\|q\|^2} \right\|_{L^2 \left( \Sigma \right) } \left\|\nabla_q \psi  \right\|_{L^2 \left( \Sigma \right) }
\end{align*}

With the help of \eqref{ee7} we further get

\beq\label{ee8}
\left|\int_\Sigma \dfrac{\psi q}{\e+1-\|q\|^2}\cdot \nabla_q \psi  \right|\leq c_H \left\|\nabla_q \psi  \right\|^2_{L^2\left(\Sigma\right)},\, \forall \psi \in H^1_0 \left( B(0,1) \right)
\eeq

Then

\begin{align*}
& b(t,\psi,\psi) + \dfrac{2}{\textbf{De}} \int_\Sigma \dfrac{\psi q}{\e+1-\|q\|^2}\cdot \nabla_q \psi \geq \dfrac{\theta_{\text{min}}}{ \textbf{De}} \int_\Sigma \left\| \nabla_x \psi  \right\|^2 \\
& - \dfrac{1}{ \textbf{De}}  \left\| \nabla_x \theta  \right\|_{L^\infty(\Omega)} \int_\Sigma \left\| \nabla_x \psi  \right\|\left\| \nabla_q \psi  \right\| + \left( \dfrac{\theta_{\text{min}}}{Q_0^2 \textbf{De}} - 2\dfrac{ c_H}{ \textbf{De}} \right) \int_\Sigma \left\| \nabla_q \psi  \right\|^2 \\
& = \left\langle Bz,z \right\rangle 
\end{align*}

where 

\begin{eqnarray*}
B= 
\left( \begin{array}{cc}
\dfrac{\theta_{\text{min}}}{ \textbf{De}} & \quad -\dfrac{\left\|\nabla_x \theta  \right\|_{L^\infty(\Omega)}}{ \textbf{De}} \\[1cm]
-\dfrac{\left\|\nabla_x \theta  \right\|_{L^\infty(\Omega)}}{ \textbf{De}} & \quad \dfrac{\theta_{\text{min}}}{Q_0^2 \textbf{De}}- \dfrac{2 c_H}{\textbf{De}}
\end{array}\right)
\end{eqnarray*}

Actually, $B$ has to be a symmetric positive definite matrix; for those features to hold true we need to assume the following necessary and sufficient condition regarding the data:

\beq\label{ee9}
\dfrac{\theta^2_{\text{min}}}{ Q_0^2}- 2 c_H \theta_{\text{min}} - \left\|\nabla_x \theta  \right\|^2_{L^\infty(\Omega)}>0
\eeq

\begin{remark}
The assumption \eqref{ee9} is stronger then that of \eqref{r5} and  is valid insofar either $Q_0$ is small enough, or $\theta_{\text{min}}$ is sufficiently large.  In the following we shall tacitly admit the assumption \eqref{ee9} to hold true.
\end{remark}

Then one deduces there exists a $\Lambda_M>0$ s.t.

\begin{align}\label{ee10}
& b\left(t,\psi,\psi \right)+ \dfrac{2}{\textbf{De}}  \int_\Sigma \dfrac{\psi q}{\e+1-\|q\|^2} \cdot \nabla_q \psi \geq \Lambda_M \|\psi\|^2_{H^1(\Sigma)},\, \forall \psi\in H^1_0(\Sigma)
\end{align}

From the above, together with \eqref{ee5}, \eqref{ee6}, one obtains

\begin{align}\label{ee11}
& \dfrac{1}{2}\dfrac{d}{dt} \int_\Sigma f_\e^2 + \Lambda_M \left\|f_\e  \right\|^2_{H^1(\Sigma)}\leq \left|  \int_\Sigma \g \cdot q f_\e \cdot \nabla_q f_\e \right|\nonumber\\
& + \dfrac{1}{ \textbf{De}}\left|  \int_\Sigma f_\e g_\e\left(f_\e \right) \cdot \nabla_x f_\e \left( 2 \nabla_x \theta + q \nabla^2_x \theta \right) \right| 
\end{align}

Observe now that for any $\delta>0$, there exists a $c(\delta)\geq0$ (independent of $\e$), s.t.

\[ z|\ln z|\leq c(\delta) +z^{1+\delta} ,\, \forall z>0   \]

Since 

\[ \left| g_\e\left(z \right) \right|\leq |\ln z|,\, \forall z>0 \]

then

\[ z \left| g_\e\left(z \right) \right|\leq  c(\delta) +z^{1+\delta} ,\, \forall z>0 \]

On the other hand now, 

\[ |z| \left| g_\e\left(z \right) \right|\leq 1,\, \forall z<0 \]

From the above it follows that: for any $\delta>0$, there exists a $c(\delta)\geq0$ (independent of $\e$), s.t.

\beq\label{ee12}
\left| z g_\e\left(z \right) \right|\leq c(\delta)+|z|^{1+\delta} ,\, \forall z\in\mb{R}
\eeq

Now, from \eqref{ee11} and capitalizing on \eqref{ee12} gives

\begin{align}\label{ee13}
& \dfrac{1}{2}\dfrac{d}{dt} \int_\Sigma f_\e^2 + \Lambda_M \left\|f_\e  \right\|^2_{H^1(\Sigma)} \leq \dfrac{\Lambda_M}{4} \left\| \nabla_q f_\e \right\|^2_{L^2(\Sigma)}+ \dfrac{1}{\Lambda_M}\|\g\|^2_{L^\infty(\Omega)}\left\|  f_\e \right\|^2_{L^2(\Sigma)} \nonumber\\
& + \dfrac{3}{ \textbf{De}} \|\theta\|_{W^{2,\infty}(\Omega)}c(\delta)\int_\Sigma \left\|\nabla_x f_\e \right\| + \dfrac{3}{ \textbf{De}} \|\theta\|_{W^{2,\infty}(\Omega)} \int_\Sigma |f_\e|^{1+\delta}\left\|\nabla_x f_\e \right\|
\end{align}

Using the fact that

\[\displaystyle \int_\Sigma \left\|\nabla_x f_\e \right\|\leq \left|\text{mes}(\Sigma)  \right|^{1/2}  \left\|\nabla_x f_\e \right\| _{L^2(\Sigma)} \]

and that

\begin{align*}
& \dfrac{3}{ \textbf{De}}c(\delta) \left|\text{mes}(\Sigma)  \right|^{1/2} \|\theta\|_{W^{2,\infty}(\Omega)}   \left\|\nabla_x f_\e \right\| _{L^2(\Sigma)} \leq \dfrac{\Lambda_M}{4} \left\|\nabla_x f_\e \right\|^2 _{L^2(\Sigma)} \nonumber\\
& + \dfrac{1}{\Lambda_M} \dfrac{9}{ \textbf{De}^2} c^2(\delta)\left|\text{mes}(\Sigma)  \right| \|\theta\|^2_{W^{2,\infty}(\Omega)} 
\end{align*}

Next, using \eqref{ee13} we further obtain

\begin{align}\label{ee14}
& \dfrac{d}{dt} \int_\Sigma f_\e^2 + \Lambda_M \left\|f_\e  \right\|^2_{H^1(\Sigma)}\leq c_4 \left\|f_\e  \right\|^2 +c_5 + \dfrac{3}{ \textbf{De}} \|\theta\|_{W^{2,\infty}(\Omega)} \int_\Sigma \left|f_\e\right|^\delta |f_\e| \left\|\nabla_x f_\e \right\|
\end{align}

By selecting now a $\delta$ s.t. $0<\delta<1/2$, one gets

\[\int_\Sigma \left|f_\e\right|^\delta |f_\e| \left\|\nabla_x f_\e \right\| \leq  \left\| | f_\e|^\delta \right\|_{L^{1/\delta}(\Sigma)} \left\| f_\e \right\|_{\displaystyle L^{\frac{2}{1-2\delta}}(\Omega)} \left\| \nabla_x f_\e \right\|_{L^2(\Omega)} \]

Since, by \eqref{ee4}

\[ \left\| | f_\e |^\delta \right\|_{L^{1/\delta}(\Sigma)} = \left\| f_\e \right\|^\delta_{L^{1}(\Sigma)} \leq \|f_0\|^{\delta}_{L^{1}(\Sigma)} \]

then

\beq\label{ee15}
\int_\Sigma \left|f_\e\right|^\delta |f_\e | \left\|\nabla_x f_\e \right\| \leq \|f_0\|^{\delta}_{L^{1}(\Sigma)} \left\| f_\e \right\|_{\displaystyle L^{\frac{2}{1-2\delta}}(\Sigma)} \left\| \nabla_x f_\e \right\|_{L^2(\Sigma)}
\eeq

Next, by Sobolev's inclusions, taking a $\delta>0$ small enough, there exists a $\delta_1\in (0,1)$ s.t. 

\[\left\| f_\e \right\|_{\displaystyle L^{\frac{2}{1-2\delta}}(\Omega)} \leq c(\delta_1)\left\|f_\e \right\|_{H^{\delta_1}(\Omega)}, \, c(\delta_1)>0 \]

By interpolation we have

\[\left\|f_\e \right\|_{H^{\delta_1}(\Omega)} \leq  c(\delta_2) \left\|f_\e \right\|^{1-\delta_1}_{L^2(\Sigma)} \left\|f_\e \right\|^{\delta_1}_{H^1(\Sigma)} \]

Now, from \eqref{ee14} and \eqref{ee15}

\begin{align*}
& \dfrac{d}{dt} \int_\Sigma f_\e^2 + \Lambda_M \left\|f_\e  \right\|^2_{H^1(\Sigma)}\leq  c_4 \left\|f_\e  \right\|^2_{L^2(\Sigma)} +c_5 + \dfrac{3}{ \textbf{De}} \|f_0\|^{\delta}_{L^{1}(\Sigma)} \left\|f_\e \right\|^{1-\delta_1}_{L^2(\Sigma)} \left\|f_\e \right\|^{1+\delta_1}_{H^1(\Sigma)}
\end{align*}

By Young's inequality and for any $\eta>0$,

\begin{align*}
& \dfrac{3}{ \textbf{De}} \|f_0\|^{\delta}_{L^{1}(\Sigma)} \left\|f_\e \right\|^{1-\delta_1}_{L^2(\Sigma)} \left\|f_\e \right\|^{1+\delta_1}_{H^1(\Sigma)} \leq \dfrac{1+\delta_1}{2} \left(\eta \left\|f_\e \right\|^{1+\delta_1}_{H^1(\Sigma)} \right)^{\frac{2}{1+\delta_1}} \\
& + \dfrac{1-\delta_1}{2} \left[ \dfrac{3}{\eta \textbf{De}} \|f_0\|^{\delta}_{L^{1}(\Sigma)} \left\|f_\e \right\|^{1-\delta_1}_{L^2(\Sigma)} \right]^{\frac{2}{1-\delta_1}} 
\end{align*}

Taking $\eta>0$ small enough gives

\[\dfrac{d}{dt} \int_\Sigma f_\e^2 + \dfrac{\Lambda_M}{2} \left\|f_\e  \right\|^2_{H^1(\Sigma)}\leq c_5 + c_6  \left\|f_\e \right\|^2_{L^2(\Sigma)}  \]

By Gronwall's inequality we deduce (proceeding in a classical manner) that there exists a constant $c>0$ (which is  independent of $\e$ but depending upon $T$) s.t.

\beq\label{ee16}
\left\|f_\e  \right\|_{L^\infty\left(0,T; L^2(\Sigma) \right) }+\left\|f_\e  \right\|_{L^2\left(0,T; H^1(\Sigma) \right) } \leq c
\eeq

From \eqref{ee16} and \eqref{ee2} we can also prove that

\beq\label{ee17}
\left\|\dfrac{d f_\e}{dt}  \right\|_{L^\infty\left(0,T; H^{-1}(\Sigma) \right)}\leq c
\eeq

Actually, to prove \eqref{ee17} above, observe that $\displaystyle \left| a_\e\left(t,f_\e,f_\e,\psi \right)\right| \leq c \left\|f_\e  \right\|_{H^1(\Sigma)} \|\psi\|_{H^1(\Sigma)}  $ due to Hardy's inequality \eqref{ee7} and to \eqref{ee12} in which we set $\delta=1$.

Eventually, from \eqref{ee16} and \eqref{ee17}, we see that 

\beq\label{ee18}
\left\|f_\e  \right\|_{X_T}\leq c
\eeq

\subsubsection{Proof Of The Non-Negativity Of $f_\e$.}\label{ng}

We make the (physically sound) assumption that

\beq\label{ng1}
f_0\geq0
\eeq

Let now $f_\e$ be expressed as

\[f_\e=f^+_\e-f^-_\e,\, f^+_\e=\text{max}\{f_\e,0 \},\, f^-_\e=-\text{min}\{f_\e,0 \}  \]

Our goal is now to prove that $f^-_\e=0$ (so that $f_\e=f^+_\e\geq 0$).  We first apply \eqref{ee2} to $f^-_\e$ and this gives

\[ \left\langle \dfrac{d f_\e}{dt},f^-_\e  \right\rangle + a_\e\left(t,f_\e,f_\e,f^-_\e \right) =0\]

We now have (by a density-type argument) that

\[\left\langle \dfrac{d f_\e}{dt},f^-_\e  \right\rangle = -\dfrac{1}{2}\dfrac{d}{dt} \left\|f^-_\e  \right\|^2_{L^2(\Sigma)} \]

Next, since $\nabla_x f^+_\e \cdot \nabla_x f^-_\e = 0$

\[ b\left(f_\e, f^-_\e \right)= b\left(f^+_\e - f^-_\e,f^-_\e  \right)=-b\left(f^-_\e, f^-_\e \right)  \]

and

\[\displaystyle \int_\Sigma \left( v\cdot \nabla_x f_\e\right)f^-_\e = - \int_\Sigma \left( v\cdot \nabla_x f^-_\e\right)f^-_\e =0   \]

Next, for any $\eta>0$,

\begin{align*}
& \dfrac{1}{ \textbf{De}} \left|  \displaystyle \int_\Sigma g_\e\left(f_\e \right)f_\e\cdot \nabla_x f^-_\e \left(2\nabla_x \theta +  \nabla^2_x \theta \cdot q \right)  \right|\nonumber\\ 
& = \dfrac{1}{ \textbf{De}} \left|  \displaystyle \int_\Sigma g_\e\left(f_\e \right)f^-_\e\cdot \nabla_x f^-_\e \left(2\nabla_x \theta +  \nabla^2_x \theta \cdot q \right)  \right|\nonumber\\  
& \leq c(\e)\left\|f^-_\e \right\|_{L^2(\Sigma)} \left\|\nabla_x f^-_\e \right\|_{L^2(\Sigma)} \leq \eta \left\|\nabla_x f^-_\e \right\|^2_{L^2(\Sigma)} + \dfrac{c^2(\e)}{4\eta}\left\|f^-_\e \right\|_{L^2(\Sigma)} 
\end{align*}

Using that $\dfrac{1}{\e+1-\|q\|^2}\leq \dfrac{1}{\e}$, then for any $\eta>0$ one also gets

\begin{align*}
& \left| - \displaystyle \int_\Sigma \g\cdot q f_\e \nabla_q f^-_\e + \dfrac{2}{\textbf{De}} \int_\Sigma \dfrac{f_\e }{\e+1-\|q\|^2}q\cdot \nabla_q f^-_\e \right|\leq \eta \left\|\nabla_x f^-_\e \right\|^2_{L^2(\Sigma)}+ c(\e,\eta)\left\|f^-_\e \right\|_{L^2(\Sigma)}
\end{align*}

With the help of inequality \eqref{r6} and by taking $\eta>0$ small enough we obtain

\[\dfrac{d}{dt}\left\|f^-_\e  \right\|^2_{L^2(\Sigma)}\leq c(\e) \left\|f^-_\e  \right\|^2_{L^2(\Sigma)}   \]

Since $f^-_\e(t=0)=0 $, use of Gronwall's inequality leads to $f^-_\e=0 $, or put it differently,

\beq\label{ng2}
f_\e\geq0
\eeq

\subsection{Performing The Limit $\e\to0$.}\label{nge}

We now state the hardcore result:

\begin{theorem}[\textbf{Main Existence Result}]
There exists at least one solution $f\in X_T$, $f\geq0$ to the problem stated in \eqref{e10}.
\end{theorem}

\begin{proof}
 
The estimates of Sections \ref{ex} and \ref{ng} allow to deduce the existence of a $f\in X_T$, $f\geq0$, s.t. we have (up to a subsequence of $\e$, for simplicity also denoted by $\e$)

\begin{align}\label{ng3}
& f_\e \xrightharpoonup[]{L^2\left(0,T; H^1_0(\Sigma) \right) }  f_\e,\, \text{weakly} \nonumber\\
& f_\e \xrightharpoonup[]{L^\infty\left(0,T; L^2(\Sigma) \right) }  f_\e,\, \text{weakly}-\ast \nonumber\\
& \dfrac{d f_\e}{dt}\xrightharpoonup[]{L^2\left(0,T; H^{-1}(\Sigma) \right) } \dfrac{df}{dt},\, \text{weakly} \nonumber\\
& f_\e \xrightarrow[]{L^2\left(\Sigma_T) \right) }  f,\, \text{strongly, by compactness.} 
\end{align}

We now pass to the limit for $\e\to0$ in the variational formulation given in \eqref{e12}, \eqref{e13}, \eqref{e14}, which is:  

\begin{align}\label{ng4}
& -\displaystyle \int_0^T \int_{\Sigma} f_\e  \psi(x,q)\f'(t) -\int_{\Sigma} f_0(x,q)\f(0)\psi(x,q) +\displaystyle \int_0^T \int_{\Sigma} \left( v\cdot \nabla_x f_\e\right) \psi\f \nonumber\\
&  + \int_0^T \int_{\Sigma} b\left(t,f_\e,\psi \right)\f +  \dfrac{1}{ \textbf{De}} \int_0^T \int_{\Sigma}f_\e g_\e\left(f_\e \right)\nabla_x \psi \cdot   \left( 2\nabla_x \theta + \nabla^2_x \theta \cdot q\right)\f \nonumber\\
& - \displaystyle \int_0^T \int_{\Sigma}\left(  \g \cdot q f_e \cdot \nabla_q \psi \right)\f +  \dfrac{2}{\textbf{De}} \displaystyle \int_0^T \int_{\Sigma} \left( \dfrac{f_\e}{\e+1-\|q\|^2}q \cdot \nabla_q \psi\right) \f=0, \nonumber\\
& \forall \psi\in H^1_0(\Sigma), \, \forall \f\in\ms{C}^1(0,T),\, \f(T)=0
\end{align}

Let $\psi\in \ms{D}(\Sigma)$ in \eqref{ng4}.  Due to the convergences stated in \eqref{ng3} we get

\begin{align}\label{ng5}
& -\displaystyle \int_0^T \int_{\Sigma} f_\e  \psi(x,q)\f'(t) +\displaystyle \int_0^T \int_{\Sigma} \left( v\cdot \nabla_x f_\e\right) \psi\f + \int_0^T b\left(t,f_\e,\psi \right)\f \nonumber\\
& - \displaystyle \int_0^T \int_{\Sigma} \left( \g\cdot q f_\e\cdot \nabla_q \psi\right)  \f \nonumber\\
& \xrightarrow[\e\to0]{} -\displaystyle \int_0^T \int_{\Sigma} f \psi \f' + \displaystyle \int_0^T \int_{\Sigma} \left( v\cdot \nabla_x f\right) \psi\f \nonumber\\
& + \int_0^T b\left(t,f,\psi \right)\f - \displaystyle \int_0^T \int_{\Sigma} \left( \g\cdot q f\cdot \nabla_q \psi\right)  \f
\end{align}

Let us now prove the convergence 

\begin{align}\label{ng6}
& \int_0^T \int_{\Sigma}E_\e\left( f_\e\right)  \nabla_x \psi \cdot    \left( 2\nabla_x \theta + \nabla^2_x \theta\cdot q\right)\f \xrightarrow[\e\to0]{} \int_0^T \int_{\Sigma}E\left( f\right)  \nabla_x \psi \cdot    \left( 2\nabla_x \theta + \nabla^2_x \theta \cdot q\right)\f 
\end{align}

To achieve this, it suffices to prove the strong convergence 

\[  \displaystyle E_\e\left( f_\e\right) \xrightarrow[\e\to0]{L^2\left(\Sigma_T \right) } E\left( f\right) \]

From the strong convergence 

\[f_\e \xrightarrow[\e\to0]{\, } f  \]

we deduce that (up to a subsequence of $\e$)

\[ f_\e(x,t) \xrightarrow[\e\to0]{L^2\left(\Sigma_T \right) } f(x,t),\, \text{a.e.}\, (x,t)\in \left( \Sigma_T\right)  \]

and

\[\left|f_\e(x,t)  \right|\leq h(x,t),\, \text{a.e.}\, (x,t)\in \Sigma_T,\, h\in L^2\left( \Sigma_T\right)  \]

Observe function $E$ is decreasing on $(0,1/e)$, and increasing on $(1/e,+\infty)$.  Then,

\[\left| E\left( f_\e(x,t)\right)  \right| \leq \dfrac{1}{e},\, \text{for}\, 0\leq f_\e(x,t) \leq 1 \]

and

\[\left| E\left( f_\e(x,t)\right)  \right| \leq E\left( h(x,t)\right) ,\, \text{for}\, f_\e(x,t) > 1 \]

and 

\[\left| E_\e \left( z \right)  \right| \leq \left| E (z) \right| ,\,  \forall z\geq 0\]

It follows that for any $\delta>0$, there exists $c(\delta)>0$ independent of $\e$, s.t.

\beq\label{ng7}
\left| E_\e \left( f_\e(x,t)\right)  \right| \leq c(\delta)  +  h^\delta\left(x,t \right) ,\,\text{a.e.}\, (x,t)\in \Sigma_T 
\eeq

Next, consider $(x,t)\in\Sigma_T$ s.t. $f(x,t)>0$.  Taking $\e$ small enough we have

\[   g_\e \left(f_\e(x,t) \right)  = \ln \left(f_\e(x,t) \right)     \]

Then, by continuity 

\[ f_\e(x,t) g_\e \left(f_\e(x,t) \right) \xrightarrow[\e\to0]{ }  f(x,t) \ln \left(f(x,t) \right) = E\left(f(x,t) \right)   \]

Let us now consider $(x,t)\in \Sigma_T$ for which $f(x,t)=0$.  Then,

\[\left|f_\e(x,t)  g_\e \left(f_\e(x,t) \right) \right| \leq \left|f_\e(x,t) \ln \left(f_\e(x,t) \right)   \right| =\left|E\left(f_\e(x,t) \right)  \right| \xrightarrow[\e\to0]{ } E\left(f(x,t) \right) =0  \]

Then, for a.e. $(x,t)\in\Sigma_T$ we have 

\[E_\e\left(f_\e(x,t) \right)  \xrightarrow[\e\to0]{ } E\left(f(x,t) \right)  \]

With the help of \eqref{ng7} and upon using Lebesgue's Dominated Convergence Theorem one gets

\[E_\e\left(f_\e \right)  \xrightarrow[\e\to0]{L^2\left(\Sigma_T \right)  } E\left(f \right),\, \text{strongly}  \]

which ends the announced proof for \eqref{ng6}.

What is left over now to prove is the validity of the convergence

\begin{equation}\label{ng8}
\displaystyle \int_0^T \int_{\Sigma} \left( \dfrac{f_\e}{\e+1-\|q\|^2}q \cdot \nabla_q \psi\right) \f \xrightarrow[\e\to0]{} \displaystyle \int_0^T \int_{\Sigma} \left( \dfrac{f}{1-\|q\|^2}q \cdot \nabla_q \psi\right) \f
\end{equation}

We have

\begin{align*}
& \left\| \dfrac{1}{\e+1-\|q\|^2} \nabla_q \psi -\dfrac{1}{1-\|q\|^2} \nabla_q \psi \right\|= \dfrac{\e}{\left(1-\|q\|^2 \right) \left(\e+1-\|q\|^2 \right) } \left\| \nabla_q \psi \right\| \\
& \leq \e \dfrac{1}{\left(1-\|q\|^2 \right)^2 }\left\| \nabla_q \psi \right\|
\end{align*}

Since $\psi \in \ms{D}(\Sigma)$, then there exists a $\delta_2>0$ s.t. $1-\|q\|^2\geq \delta_2 $ whenever $\nabla_q \psi(x,q)\neq 0 $.  Then the function 

\[ \dfrac{1}{\left(1-\|q\|^2 \right)^2 }\left\| \nabla_q \psi \right\| \in L^\infty(\Sigma) \]

Therefore

\[ \dfrac{1}{\e+1-\|q\|^2} \nabla_q \psi  \xrightarrow[\e\to0]{L^\infty(\Sigma)} \dfrac{1}{1-\|q\|^2  } \nabla_q \psi,\, \text{strongly} \]

which implies the statement in \eqref{ng8}.  We then obtain from \eqref{ng4}, \eqref{ng5}, \eqref{ng6}, and \eqref{ng8}, that for any $\psi \in \ms{D}(\Sigma)$ and for any $\f\in \ms{C}^1\left(0,T \right)$, s.t. $\f(T)=0 $,

\begin{align}\label{ng9}
& -\displaystyle \int_0^T \int_{\Sigma} f \psi\f' -\int_{\Sigma} f_0 \f(0)\psi  +\displaystyle \int_0^T \int_{\Sigma} \left( v\cdot \nabla_x f\right) \psi\f \nonumber\\
&  + \int_0^T\int_{\Sigma} b\left(t,f,\psi \right)\f +  \dfrac{1}{ \textbf{De}} \int_0^T \int_{\Sigma}E(f)\nabla_x \psi \cdot   \left( 2\nabla_x \theta + q\nabla^2_x \theta\right)\f \nonumber\\
& - \displaystyle \int_0^T \int_{\Sigma}\left(  \g \cdot q f \cdot \nabla_q \psi \right)\f +  \dfrac{2}{\textbf{De}} \displaystyle \int_0^T \int_{\Sigma} \left( \dfrac{f}{1-\|q\|^2}q \cdot \nabla_q \psi\right) \f=0
\end{align}
 
Remark that $E(f)\in L^\infty\left(0,T;L^2(\Omega) \right)$ (because $\left|E(z) \right|\leq c(1+z^2) $, for any $z\geq0$), and that $\displaystyle \dfrac{f}{1-\|q\|^2}\in L^2\left( \Sigma_T\right)$ (because $f\in L^2\left(0,T; H^1_0(\Sigma) \right)$ and due to Hardy's inequality).  Then, because $\ms{D}(\Sigma)$  is densely included into $H^1_0(\Sigma)$, it allows to deduce that \eqref{ng9} is also valid for any $\psi\in H^1_0(\Sigma)$, fact that ends the proof.
 
\end{proof}

\section{Conclusions.}

In this paper we extended the early work of Curtiss and Bird ~\cite{bird3} on kinetic theory describing the temperature influence on the dynamics of polymer fluids.  Specifically, we derived a new configurational probability equation without the originally proposed ``linear gradient approximation'', fact that may account for some shortcomings pointed out in the original work ~\cite{bird3}.  The resulting transport equation of this  work is consequently non-linear and (hence) more complex in nature compared to the original (linear) one.   As a first step towards putting it to work for practical purposes we proved the existence of positive variational solutions.

Subsequent work devoted to obtaining the corresponding temperature diffusion equation by
accounting for the various molecular interactions responsible for the heat that is conveyed by
the fluid is the focus of a forthcoming paper.

\section{Acknowledgements} 


The Authors were deeply saddened by the untimely passing of Professor Genevi\`eve Raugel whom they got to know and took benefit of her outstanding scientific skills and, also, to appreciate her generous human and personality stances.   

L. I. Palade gratefully acknowledges Professor C\u at\u alin Radu Picu, RPI, Troy (NY), for useful talks on polymer dynamics.

\end{document}